\documentclass[11pt]{article}
\usepackage[margin=1in]{geometry}

\usepackage[usenames,dvipsnames]{xcolor}
\definecolor{Gred}{RGB}{219, 50, 54}
\definecolor{ToCgreen}{RGB}{0, 128, 0}

\usepackage[T1]{fontenc}

\usepackage[scale=0.97]{XCharter} 

\usepackage[libertine,bigdelims,vvarbb,scaled=1.05]{newtxmath} 

\usepackage{babel}
\usepackage[spacing=true,kerning=true,babel=true,tracking=true]{microtype}

\DeclareMathAlphabet{\pazocal}{OMS}{zplm}{m}{n} 
\renewcommand{\mathcal}[1]{\pazocal{#1}}

\usepackage{graphicx}
\usepackage{bm}
\usepackage{amsmath}
\usepackage{amsfonts}
\usepackage{amsthm}
\usepackage{mathtools}
\usepackage{hyperref}
\hypersetup{
      colorlinks=true,
  citecolor=ToCgreen,
  linkcolor=Sepia,
  filecolor=Gred,
  urlcolor=Gred
  }
\usepackage{braket}
\usepackage[normalem]{ulem}
\usepackage{wrapfig}
\usepackage{tikz}
\usepackage{dsfont}
\usepackage{comment}
\usepackage{algpseudocode}
\usepackage{algcompatible}
\usepackage{thmtools,thm-restate}
\mathtoolsset{showonlyrefs}
\usepackage{newfloat}
\usepackage{enumitem}

\hypersetup{colorlinks=true,urlcolor=[rgb]{0,0,0.5},citecolor=[rgb]{0.5,0,0},linkcolor=[rgb]{0,0,0.4}}

\usetikzlibrary{decorations.pathmorphing}

\newcommand{\calD}{\mathcal{D}}

\newcommand{\calS}{\mathcal{S}}

\newcommand{\calZ}{\mathcal{Z}}
\renewcommand{\epsilon}{\varepsilon}

\newcommand{\Z}{\mathbb{Z}}

\newcommand{\C}{{\mathbb{C}}}

\renewcommand{\vec}[1]{\mathbf{#1}}
\newcommand{\wh}{\widehat}

\newcommand{\Split}{\mathsf{Split}}
\newcommand{\Dsplit}{\mathsf{DSplit}}

\newcommand{\norm}[1]{\left\lVert#1\right\rVert}

\newtheorem{theorem}{Theorem}[section]
\newtheorem{lemma}[theorem]{Lemma}
\newtheorem{definition}[theorem]{Definition}
\newtheorem{corollary}[theorem]{Corollary}

\newtheorem{fact}[theorem]{Fact}
\newtheorem{remark}[theorem]{Remark}
\newtheorem{claim}[theorem]{Claim}

\newcommand{\poly}{\mathrm{poly}}
\newcommand{\tr}{\mathrm{tr}}
\newcommand{\eps}{\epsilon}
\newcommand{\diag}{\mathrm{diag}}
\newcommand{\op}{\textsf{op}}

\DeclarePairedDelimiter{\iprod}{\langle}{\rangle}
\DeclarePairedDelimiter{\brk}{[}{]}

\DeclareMathOperator{\E}{\mathbb{E}}
\DeclareMathOperator{\Var}{\mathsf{Var}}
\newcommand{\Cov}{\mathsf{Cov}}

\newcommand{\Sp}{\mathsf{Sp}}

\usepackage[size=small]{caption}
\usepackage{etoolbox}

\usepackage{algorithm}
\usepackage{algpseudocode}

\begin{document}

\title{Optimal high-precision shadow estimation}

\author{
Sitan Chen
\thanks{SEAS, Harvard University. Email: \href{mailto:sitan@seas.harvard.edu}{sitan@seas.harvard.edu};}
\qquad\qquad
Jerry Li
\thanks{Microsoft Research. Email: \href{mailto:jerrl@microsoft.com}{jerrl@microsoft.com};}
\qquad\qquad
Allen Liu
\thanks{MIT. Email: \href{mailto:cliu568@mit.edu}{cliu568@mit.edu}; This work was supported in part by an NSF Graduate Research Fellowship and a Fannie and John Hertz Foundation Fellowship}
}

\maketitle

\begin{abstract}
    We give the first tight sample complexity bounds for shadow tomography and classical shadows in the regime where the target error is below some sufficiently small inverse polynomial in the dimension of the Hilbert space. Formally we give a protocol that, given any $m\in\mathbb{N}$ and $\epsilon \le O(d^{-12})$, measures $O(\log(m)/\epsilon^2)$ copies of an unknown mixed state $\rho\in\mathbb{C}^{d\times d}$ and outputs a classical description of $\rho$ which can then be used to estimate any collection of $m$ observables to within additive accuracy $\epsilon$. Previously, even for the simpler task of shadow tomography \--- where the $m$ observables are known in advance \--- the best known rates either scaled benignly but suboptimally in all of $m, d, \epsilon$~\cite{buadescu2021improved,watts2024quantum}, or scaled optimally in $\epsilon, m$ but had additional polynomial factors in $d$ for general observables~\cite{huang2020predicting}. Intriguingly, we also show via dimensionality reduction, that we can rescale $\eps$ and $d$ to reduce to the regime where $\epsilon \le O(d^{-1/2})$. Our algorithm draws upon representation-theoretic tools recently  developed in the context of full state tomography~\cite{chen2024optimal}.
\end{abstract}

\newpage

\tableofcontents

\newpage

\section{Introduction}

In this paper, we consider the well-studied and related problems of shadow tomography~\cite{aaronson2018shadow,aaronson2019gentle,buadescu2021improved,watts2024quantum} and learning classical shadows~\cite{huang2020predicting,huang2022learning,grier2022sample}.
In both problems, there is an unknown quantum state, and the goal is to simultaneously estimate as many linear properties of this state as possible, using as few copies of $\rho$ as possible.
Such problems come up naturally in a variety of real-world laboratory settings, see e.g.~\cite{peruzzo2014variational,bluvstein2024logical,bentsen2019integrable,mcclean2016theory,kandala2017hardware}.
The power of these frameworks is that these shadow estimation tasks can be performed \emph{efficiently}. 
That is, to simultaneously predict $m$ properties of the state, one only requires $\poly (\log m)$ many copies of the state, and $\poly (\log m)$ time.
Indeed, algorithms for classical shadows have already shown immense potential in real-world evaluations~\cite{zhang2021experimental,struchalin2021experimental,levy2024classical}.

Formally, we let $\rho \in \C^{d \times d}$ be an arbitrary and unknown $d$-dimensional mixed state.
In shadow tomography, we are given a set of $m$ linear observables $\{O_i\}_{i = 1}^m \subset \C^{d \times d}$ satisfying $0 \preceq O_i \preceq I$, and given $n$ copies of $\rho$, our goal is to estimate $\tr (O_i \rho)$ to accuracy $\eps$, for all $i = 1, \ldots, m$, with high probability.

In classical shadows, the setup is the same, except that the measurements must be chosen obliviously with respect to the collection of observables $\{O_i\}_{i = 1}^m$.
An algorithm for learning classical shadows can be broken down into two phases.
First, in the measurement phase, the algorithm performs measurements on $n$ copies of $\rho$, to obtain a classical representation of $\rho$, which is referred to the classical shadow of $\rho$.
Then, in the estimation phase, the algorithm is given $m$ observables $\{O_i\}_{i = 1}^m$, and it must output estimates of $\tr (O_i \rho)$ to accuracy $\eps$ for all $i = 1, \ldots, m$, based solely on the observables and the classical shadow of $\rho$.

Despite significant research interest in the area, prior to our work, the exact complexity of these tasks for general $\rho$ and $\{O_i\}^m_{i=1}$ was unknown, for any nontrivial regime of $m, d, \epsilon$.
For shadow tomography, the best known rate is due to~\cite{buadescu2021improved,watts2024quantum}, who give an algorithm that uses 
\begin{equation}
n = O\left(\frac{\log^2(m)\cdot \log(d)}{\eps^4} \right)
\end{equation}
copies of $\rho$, whereas the best known lower bound is the ``trivial'' bound of $n = \Omega(\tfrac{\log m}{\eps^2})$ from the classical setting.
For classical shadows, the best known rate is due to~\cite{huang2020predicting}, who give an algorithm that (for general observables) requires $n = O(\tfrac{d \log m}{\eps^2})$ copies of $\rho$. 
The best lower bound for classical shadows is due to~\cite{grier2022sample}, who obtain a lower bound of
\begin{equation}
n = \Omega \Bigl( \Bigl( \frac{\sqrt{d}}{\eps} + \frac{1}{\eps^2} \Bigr) \log m \Bigr) \; ,
\end{equation}
which holds even when $\rho$ is pure.
The lack of tight rates for these problems is in stark contrast to other quantum learning settings such as full state tomography~\cite{o2016efficient,haah2016sample,chen2023does}, state certification~\cite{o2015quantum,chen2022tight}, and Pauli channel estimation~\cite{chen2022quantum,chen2023efficientPauli}, where asymptotically tight sample complexities are known.

In this work, we initiate the study of these problems in what we call the \emph{high-accuracy} regime, that is, when $\eps = O(d^{-c})$ for some $c$ sufficiently large.
This is in contrast to much prior theoretical work on these problems, which primarily focused on the ``low-accuracy'' regime where $m$ and $d$ are large compared to $\eps^{-1}$.
Our primary interest in this setting is two-fold.
First, from a practical point of view, this is an important setting: in practice, it is very often the case that we wish to obtain detailed information about relatively small quantum systems~\cite{zhang2021experimental,struchalin2021experimental,hadfield2022measurements}.
Second, from a mathematical point of view, we find that these problems exhibit new and very interesting properties within this regime.

Indeed, informally stated, our main algorithmic result (see Theorem~\ref{thm:main-informal}) is a new, statistically optimal estimator in this regime for both classical shadows and shadow tomography for general observables.
To our knowledge, this is the first time that tight rates have been established for these problems, in any nontrivial regime of $m, d, \epsilon$.
Qualitatively speaking, our results uncover the following, previously unknown phenomena for these problems in the high-accuracy regime:
\begin{itemize}[leftmargin=*]
    \item {\bf Shadow estimation no harder than classical counterpart.} 
    In the high-accuracy regime, we show that the rate for shadow estimation matches the corresponding ``trivial'' lower bound in the classical case, where the observables are diagonal matrices.
    To our knowledge, this is the first time where tight rates have been demonstrated for either problem, and also the only known regime where the quantum and classical rates match for general observables.
    \item {\bf Oblivious protocols can match adaptive ones.} We show that in this regime, shadow tomography and classical shadows \emph{are statistically equivalent}. In other words, there is no statistical advantage for the measurements to be chosen adaptively based on the set of linear observables of interest.
    This is perhaps surprising, as all previously known statistically efficient algorithms for shadow tomography crucially required measurements to be chosen based on the set of linear observables of interest.
    \item {\bf Representation theory for shadow estimation.} From a technical point of view, one interesting aspect of our work is that it deviates heavily from previous techniques on shadow estimation, and instead builds upon representation theoretic techniques previously used for full state tomography~\cite{chen2024optimal}. 
    This adds to a growing literature on the power of such representation theoretic tools for shadow estimation~\cite{grier2022sample,grier2024improved}; however, we emphasize that beyond this similarity, our techniques are quite distinct from these works.
\end{itemize}
Finally, we also demonstrate a formal reduction to the ``medium-accuracy'' regime (see Lemma~\ref{lem:reduce-dimension}).
That is, we show that without loss of generality, for the task of classical shadow estimation, one may assume that $\eps \leq O(d^{-1/2})$.
This reduction works by demonstrating that by leveraging classical ideas from dimensionality reduction~\cite{johnson1986extensions}, one can linearly trade off $d$ and $\eps$ in the sample complexity, as long as $\eps \geq O(d^{-1/2})$.
While this still leaves a gap in the parameter landscape where we do not know the correct rates, as our analysis currently requires $\eps \leq O(d^{-12})$, at the very least this demonstrates that all of the interesting action for this problem is in the setting where $\eps$ and $d^{-1}$ are polynomially related.
In fact, we conjecture that the correct rate for classical shadows is
\begin{equation}
\label{eq:conjecture}
\Theta \left( \left( \frac{d}{\eps} + \frac{1}{\eps^2} 
\right) \cdot \log m \right) \; ,
\end{equation}
as we conjecture that (1) the reduction holds up to the threshold of $\eps \geq \Omega (d^{-1})$, and (2) our rate of $O(\log m / \eps^2)$ holds as long as $\eps \leq O(d^{-1})$.
However, it seems that improving the thresholds on both fronts requires additional ideas, and we leave establishing this rate as interesting future work.


\section{Our contributions} 

\begin{figure}[h!]
    \centering
    \includegraphics[width=0.9\textwidth]{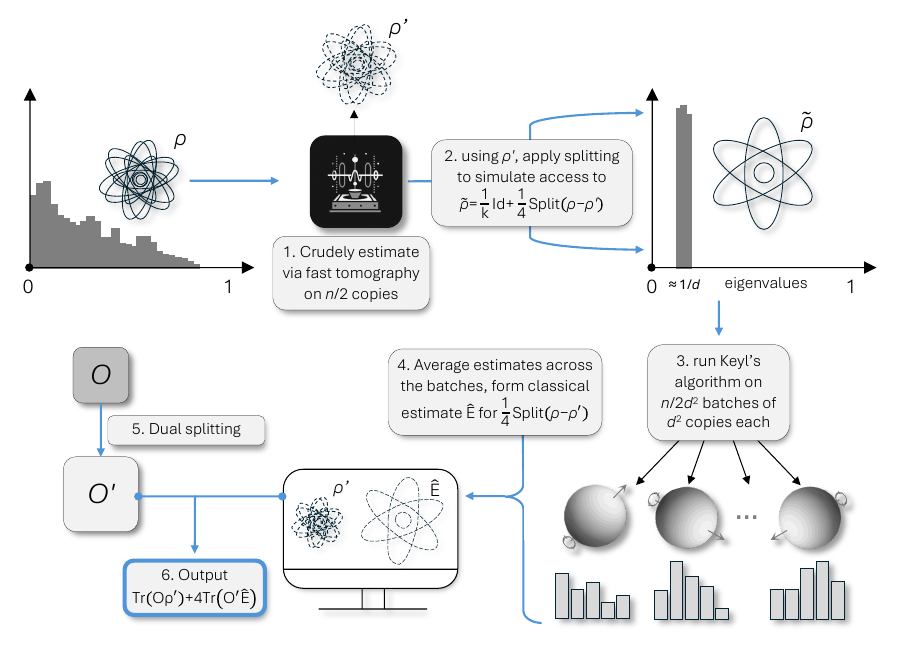}
    \caption{Overview of our shadow estimation algorithm}
    \label{fig:main}
\end{figure}

In this work, we give new algorithms for shadow tomography and for learning classical shadows.
We demonstrate that these algorithms obtain optimal sample complexities for both problems in the high-accuracy regime. 
In fact, our rates match the ``trivial'' lower bounds for these problems, demonstrating that in this regime, the quantum task is no harder than the classical one.
To our knowledge, this is the first time where tight rates have been demonstrated for either problem.

Perhaps surprisingly, we do this by giving a new algorithm for learning classical shadows which matches the lower bound for the ostensibly easier task of shadow tomography.
That is, we show that in the high-accuracy regime, learning classical shadows is no harder than shadow tomography.
More formally, we show:
\begin{theorem}
    \label{thm:main-informal}
    Let $m, d \in \Z$ be fixed, and let $\eps = O(d^{-12})$. Then, there is an estimator which takes $n$ copies of an unknown mixed state $\rho \in \C^{d \times d}$, where
    \begin{equation}
        n = O \left( \frac{\log m}{\eps^2} \right) \; ,    
    \end{equation}
    and outputs a classical function $F: \C^{d \times d} \to [0, 1]$ so that for any fixed collection of $m$ observables $\{O_i\}_{i = 1}^m$, the function satisfies $| F(O_i) - \tr (O_i \rho) | \leq \eps$ for all $i = 1, \ldots, m$ with high probability.
\end{theorem}
We pause here to make a couple of remarks on this result.
First, as alluded to above, this clearly also implies the same upper bound for shadow tomography, and moreover, this rate is tight, as it matches the lower bound for shadow tomography (and estimating statistical queries).

Second, our rate is dimension independent.
This is in contrast to prior rates for learning classical shadows, and indeed, even the lower bound~\cite{grier2022sample}.
Note that our result does not violate this lower bound because the dimension-dependent term vanishes in the high-accuracy regime, and indeed our result suggests that for classical shadows the optimal dependence on $d$ ought to be a lower order term in $\eps^{-1}$.

Thirdly, we obtain the optimal quadratic scaling in $1/\epsilon$, whereas the aforementioned upper bounds for shadow tomography scale with $1/\epsilon^4$. A recent work~\cite{chen2024optimal2} shows that with $\mathrm{polylog}(d)$-copy measurements, $\Omega(1/\epsilon^4)$ copies are necessary even in the special case when the observables are Pauli operators. We show that if the number of copies that can be measured at once scales polynomially in the dimension, then this lower bound no longer applies.

\vspace{0.5em}

\noindent \textbf{A reduction to the medium-accuracy regime} As mentioned preivously, we also demonstrate the following reduction to the medium-accuracy regime:  
\begin{theorem}[informal, see Theorem~\ref{lem:reduce-dimension}]
\label{thm:reduction-informal}
    Suppose that for all $k, \eps'$ satisfying $\eps' \leq k^{-1/2}$, there is an algorithm that solves classical shadows for $k$-dimensional states to error $\eps'$ (and constant failure probability) with $f(k, \eps')$ copies.
    Let $d, \eps$ satisfy $\eps \geq 100/\sqrt{d}$.
    Then, there is an algorithm that solves classical shadows for $d$-dimensional states to error $\eps$ (and constant failure probability) that uses $f(100 \sqrt{d} / \eps, 1 / \sqrt{d})$ copies. 
\end{theorem}
Here, we briefly pause to show how this translates to a reduction to the medium-accuracy regime.
The theorem states that to solve classical shadows to error $\eps$ for $d$ dimensional states, it suffices to obtain an estimator for classical shadows in $k = 100 \sqrt{d} / \eps$ dimensions to error $\eps' = 1 / \sqrt{d}$.
Notice that by the choice of parameters, we have that $k \leq d$, and thus consequently $\eps' \leq d^{-1/2} \leq k^{-1/2}$, so this new problem is indeed in the medium-accuracy regime.
In fact, in Section~\ref{sec:reduction} we show a more general reduction which allows us to linearly trade off $d$ and $\eps$.

Moreover, this reduction indeed recovers a linear tradeoff between $\eps$ and $d$ that we believe is optimal.
For instance, if we believe the conjectured rate~\eqref{eq:conjecture} holds for all $\eps \leq d^{-1/2}$
then a straightforward computation demonstrates that this reduction yields the same rate holds for $\eps \geq 100 d^{-1/2}$ as well.

\section{Technical overview}


\noindent Our approach is a departure from the aforementioned approaches to shadow tomography, which automatically lose an extra $\log d / \epsilon^2$ factor because they involve an outer routine based on online learning. Instead, our starting point is the recent approach of~\cite{chen2024optimal} for full \emph{state tomography}.


\vspace{0.5em}\noindent \textbf{Overview of~\cite{chen2024optimal}.} Roughly speaking, the approach in that work consisted of two components: 1) a reduction from tomography of arbitrary mixed states to tomography of mixed states which are a small perturbation of the maximally mixed state, and 2) an analysis of Keyl's estimator for learning the \emph{perturbation} that gives rise to such a state, rather than learning the state directly.

Let us first focus on step 2). We say call states that are small perturbations of $I_d/d$ \emph{balanced states}.  If we express an unknown balanced state $\rho$ as $\rho = I_d / d + E$ for some perturbation $E$, then they observed that the tensor product of $t$ copies of $\rho$ is close to the ``linearized state''
\begin{equation}
    \rho' \triangleq (I_d/d)^{\otimes t} + \sum_{{\rm sym}} E \otimes (I_d/d)^{\otimes t-1}    
\end{equation}
where $\sum_{{\rm sym}}$ denotes the sum over all tensor products of one copy of $E$ and $t - 1$ copies of $I_d/d$. Technically this linearized state is not necessarily a density matrix, but for any POVM $\{M_z\}_{z\in\mathcal{Z}}$, one can still consider measurement statistics of the form $\int_{\mathcal{Z}} f(M_z) \iprod{\rho', M_z}\,\mathrm{d}z$, which correspond to the expectation of any estimator $f$ applied to the outcome of ``measuring'' the linearized state.

Now consider a natural choice of $f$ in the context of state tomography: Keyl's estimator. Whereas the expected result of applying Keyl's estimator to the actual state $\rho$ is hard to characterize, the upshot of working with the linearization $\rho'$ is that it is much more amenable to calculations, and we can actually explicitly compute the analogous result for $\rho'$. In particular, by taking $f$ and $\{M_z\}_{z\in\mathcal{Z}}$ to be given by Keyl's estimator, the above measurement statistic $\int_{\mathcal{Z}} f(M_z) \iprod{\rho', M_z}\,\mathrm{d}z$ turns out to be exactly given by a perturbation of the maximally mixed state by some \emph{known multiple} of $E$ (Corollary~\ref{coro:fake-estimator-mean}), i.e. $I_d/d + cE$ for some known factor $c$. 

This means that if $E$ has sufficiently small norm, the expected result $\E[\wh{\rho}]$ of applying Keyl's estimator to the actual state $\rho$ is sufficiently close to this (Lemma~\ref{lem:estimator-error}) that if we had access to $\E[\wh{\rho}]$, we could simply estimate the perturbation $E$ via $c^{-1}(\E[\wh{\rho}] - I_d/d)$. Of course in reality we only have access to \emph{realizations} of the state $\wh{\rho}$ obtained by Keyl's estimator, rather than their expectation, but using existing bounds on the variance of Keyl's estimator~\cite{o2016efficient}, we can control the deviation between $\wh{\rho}$ and $\E[\wh{\rho}]$.

\vspace{0.5em}\noindent \textbf{Adapting to the shadow estimation setting.} In this work, we observe that even though the above analysis was originally implemented in~\cite{chen2024optimal} to bound the accuracy of the estimate $\wh{E}$ for perturbation $E$ obtained in Frobenius norm, essentially the same analysis translates naturally to bounding accuracy as quantified by how close $\iprod{O,E}$ is to $\iprod{O,\wh{E}}$ for an arbitrary observable $O$. The only place where we need to be somewhat careful is in controlling the variance of the estimator: instead of directly bounding the expected squared distance between $\wh{\rho}$ and $\E[\wh{\rho}]$, we need to bound the expected squared discrepancy $\E[\iprod{O, \wh{\rho} - \E[\wh{\rho}]}^2]$. While it may be tempting to simply bound this by applying Cauchy-Schwarz and appealing to the existing bound on $\E[\norm{\wh{\rho} - \E[\wh{\rho}]}^2_F]$, this is lossy by dimension-dependent factors. Instead, we need to exploit the rotation-invariance of Keyl's estimator to get a tighter bound on this variance (see Eq.~\eqref{eq:boundvariance} onwards).

The above discussion is already sufficient to prove our main result in the special case where $\rho$ is balanced. It gives rise to an estimate for $E$, and thus for $\rho$, which is \emph{oblivious} in the sense that it does not depend on the choice of observable $O$ above. Furthermore, our algorithm only needs $O(\log(1/\delta)/\epsilon^2)$ samples to produce an estimate $\wh{\rho}$ for which $|\iprod{O, \wh{\rho} - \rho}| \le \epsilon$ with probability $1 - \delta$. As $\wh{\rho}$ is oblivious to $O$, it can be used to estimate any set of $m$ observables $O_1,\ldots,O_m$ with probability $1 - m\delta$. By taking $\delta = O(1/m)$, we obtain Theorem~\ref{thm:main-informal} in the special case of balanced states. We then appeal to part 1) of the analysis in~\cite{chen2024optimal}, which allows us to effectively reduce from the case of arbitrary mixed states to the case of balanced states.

\vspace{0.5em}\noindent \textbf{Reduction to the balanced case.} Here we summarize how we adapt their approach in this step to our shadow estimation setting. At the end we comment on how it compares to the implementation in~\cite{chen2024optimal} for full state tomography.

Roughly speaking, this part of the proof is based on a certain ``splitting'' operation $\mathsf{Split}$ (see Definition~\ref{def:split}) that linearly maps any $d$-dimensional mixed state to an $O(d)$-dimensional one whose eigenvalues are upper bounded by $1/d$. Importantly, given measurement access to $\rho^{\otimes t}$, one can simulate measurement access to $\mathsf{Split}(\rho)^{\otimes t}$, and furthermore there is a \emph{dual operation} $\mathsf{DSplit}$ (see Definition~\ref{def:dsplit}) that can be applied to any observable $O$ such that the expectation value $\iprod{O,\rho}$ is equal to $\iprod{\mathsf{DSplit}(O), \mathsf{Split}(\rho)}$.

We can then reduce to the case that $\rho$ is balanced as follows: we obtain a crude estimate $\tilde{\rho}$ for $\rho$ using \emph{low-accuracy state tomography} (Theorem~\ref{thm:untentangled-tomography}), and then we ``recenter'' $\mathsf{Split}(\rho)$ around $\mathsf{Split}(\tilde{\rho})$ to get an $O(d)$-dimensional state which is sufficiently close to maximally mixed. Importantly, using the description of $\tilde{\rho}$ and by simulating measurement access to $\mathsf{Split}(\rho')$, we can simulate measurement access to this recentered state and thus reduce to the case where the unknown state is balanced (in $O(d)$ dimensions).

We remark that the primary difference between our implementation of this splitting technique and the one in~\cite{chen2024optimal} is our use of $\mathsf{DSplit}$, which is specific to the shadow estimation setting. For state tomography, \cite{chen2024optimal} considered a different operation. Roughly speaking, they defined a procedure $\mathsf{Rec}$ which \emph{inverts} the mapping given by $\mathsf{Split}$, so that once one has an estimate for the perturbation corresponding to the recentered state, their final estimator is given by applying $\mathsf{Rec}$ to this estimate. In contrast, because our goal is not to estimate the state in Frobenius norm, but rather to estimate it well enough to answer expectation value queries, we need an operation \emph{dual} to $\mathsf{Split}$ instead of an operation \emph{inverse} to it. This leads us to consider the operation $\mathsf{DSplit}$ sketched above.


\section{Outlook}

In this work we gave the first algorithm for classical shadows to achieve optimal sample complexity $O(\log m / \epsilon^2)$ for general observables in some nontrivial regime, namely when the target accuracy $\epsilon$ is inverse polynomial in the dimension $d$ of the Hilbert space. In contrast, prior work either suffered from extraneous logarithmic factors in $m$ and $d$ and polynomial factors in $1/\epsilon$, or required polynomial factors in $d$ . Interestingly, our proof leverages ideas from the recent work of~\cite{chen2024optimal} on \emph{full} state tomography. The central idea is to formulate an (approximately) unbiased estimate of the state by analyzing the behavior of Keyl's estimator on a certain linearization of the batch of copies of the unknown state.

The natural question left open by our work is to handle the low-accuracy regime, that is, to lift the assumption that $\epsilon \le 1/\mathrm{poly}(d)$. Unfortunately in the low-accuracy regime, the linearization trick mentioned above no longer applies. Resolving this would settle the main open question of~\cite{aaronson2018shadow}, i.e. showing that in all parameter regimes, the sample complexity of shadow tomography is no worse than that of its classical analogue.

Another interesting direction for future work is to understand how the sample complexity of classical shadows changes under additional constraints on $\rho$, e.g. if it has low rank or is preparable with a shallow quantum circuit. In the special case where $\rho$ is rank-1, this was settled in the work of~\cite{grier2022sample}. The recent work of~\cite{grier2024improved} obtained an improved upper bound for general low-rank states compared to the result of~\cite{huang2020predicting}.

\vspace*{6pt}
\noindent {\bf Acknowledgments.}\quad 

The authors thank Ainesh Bakshi, Jaume de Dios Pont, Ryan O'Donnell, and Ewin Tang for illuminating discussions about shadow tomography.

\bibliographystyle{plain}
\bibliography{bibliography}

\newpage

\appendix

\section{Representation Theory and Keyl's POVM}

Our algorithm will be based on Keyl's POVM \cite{keyl2006quantum} which is at the heart of most quantum state tomography algorithms that use entangled measurements \cite{wright2016learn,o2016efficient,o2017efficient,chen2024optimal}.  Defining Keyl's POVM requires some basic concepts from representation theory. 
In the following, we list the ones relevant to our discussion here; see \cite{wright2016learn} for a more detailed exposition.

\begin{definition}\label{def:young-tableaux}[Young Tableaux]
We have the following standard definitions:
\begin{itemize}
\item Given a partition $\lambda \vdash n$, a Young diagram of shape $\lambda$ is a left-justified set of boxes arranged in rows, with $\lambda_i$ boxes in the $i$th row from the top.
\item  A standard Young tableaux (SYT) $T$ of shape $\lambda$ is a Young diagram of shape $\lambda$ where each box is filled with some integer in $[n]$ such that the rows are strictly increasing from left to right and the columns are strictly increasing from top to bottom.
\item A semistandard Young tableaux (SSYT) $T$ of shape $\lambda$ is a Young diagram of shape $\lambda$ where each box is filled with some integer in $[d]$ for some $d$ and the rows are weakly increasing from left to right and the columns are strictly increasing from top to bottom.
\end{itemize}
\end{definition}

\noindent We recall the correspondence between Young tableaux and representations of the symmetric and general linear groups:

\begin{definition}
We say a representation $\mu$ of $GL_d$ over a complex vector space $\C^m$ is a polynomial representation if for any $U \in \C^{d \times d}$, $\mu(U) \in \C^{m \times m}$ is a polynomial in the entries of $U$.
\end{definition}

\begin{fact}[\cite{sagan2013symmetric}]
The irreducible representations of the symmetric group $S_n$ are exactly indexed by the partitions $\lambda \vdash n$ and have dimensions $\dim(\lambda)$ equal to the number of standard Young tableaux of shape $\lambda$.  We denote the corresponding vector space $\Sp_{\lambda}$.
\end{fact}

\begin{fact}[\cite{goodman2009symmetry}]
For each $\lambda \vdash n$, there is a (unique) irreducible polynomial representation of $GL_d$ corresponding to $\lambda$. We denote the corresponding map and vector space $(\pi_{\lambda}, V_{\lambda}^d)$.  The dimension $\dim(V_{\lambda}^d)$ is equal to the number of semistandard Young tableaux of shape $\lambda$ with entries in $[d]$.  This representation, restricted to $U_d$ is also an irreducible representation.
\end{fact}

\begin{theorem}[Schur-Weyl Duality \cite{goodman2009symmetry}]\label{thm:schur-weyl}
Consider the representation of $S_n \times GL_d$ on $(\C^{d})^{\otimes n}$ where the action of the permutation $\pi \in S_n$ permutes the different copies of $\C^{d}$ and the action of $U \in GL_d$ is applied independently to each copy.  This representation can be decomposed as a direct sum 
\[
(\C^{d})^{\otimes n} = \bigoplus_{\substack{\lambda \vdash n \\ \ell(\lambda) \leq d  }} \Sp_{\lambda} \otimes V_{\lambda}^d \,.
\]
\end{theorem}

\begin{definition}[Schur Subspace]
We call $\Sp_{\lambda} \otimes V_{\lambda}^d$ the $\lambda$-Schur subspace.  Given integers $n,d$ and $\lambda \vdash n$, we define $\Pi_{\lambda}^d: (\C^{d})^{\otimes n} \rightarrow  \Sp_{\lambda} \otimes V_{\lambda}^d $ to project onto the $\lambda$-Schur subspace.  
\end{definition}

\begin{theorem}[Gelfand-Tsetlin Basis \cite{goodman2009symmetry}]\label{thm:gf-basis}
Let $n,d$ be positive integers.  For each partition $\lambda \vdash n$ where $\lambda$ has at most $d$ parts, there is a basis $v_1, \dots , v_m$ of $V_{\lambda}^d$ with $m = \dim(V_{\lambda}^d)$ such that for any matrix $D_{\alpha} = \diag(\alpha_1, \dots , \alpha_d)$, we have $v_i^\dagger \pi_{\lambda}(D_{\alpha})v_i = \alpha^{f^{(i)}}$ for all $i$ where $f^{(i)}$ are each $d$-tuples that give the frequencies of $1,2, \dots , d$ in each of the different semi-standard tableaux of shape $\lambda$.
\end{theorem}

\begin{definition}[Maximal-weight Vector]
For a partition $\lambda \vdash n$, we define the maximal weight vector $v_{\lambda} \in V_{\lambda}^d$ to be the vector given by Theorem~\ref{thm:gf-basis} with $f^{(i)} = \lambda_i$ for all $i$.  
\end{definition}
\begin{remark}
Note that for a given partition $\lambda = (\lambda_1 , \dots , \lambda_d)$, there is a semi-standard tableaux of shape $\lambda$ with frequencies  $\lambda_1, \dots , \lambda_d$ (where we just fill the $i$th row with all entries equal to $i$) so the vector defined above indeed exists.
\end{remark}

\begin{definition}[Weak Schur Sampling]
We use the term weak Schur sampling to refer to the POVM on $\C^{d^n \times d^n}$ with elements given by $\Pi_\lambda^d$ for $\lambda$ ranging over all partitions of $n$ into at most $d$ parts.
\end{definition}

\begin{definition}[Keyl's POVM \cite{keyl2006quantum}]\label{def:keyl-povm}
We define the following POVM on $\C^{d^n \times d^n}$:  first perform weak Schur sampling to obtain $\lambda \vdash n$.  Then discard the permutation register (corresponding to the subspace $\Sp_{\lambda} $). 
 Within the remaining subspace $V_{\lambda}^d$, measure according to
\[
\dim(V_{\lambda}^d)  \{ \pi_{\lambda}(U) v_{\lambda}v_{\lambda}^\dagger \pi_{\lambda}(U)^{\dagger} \}_{U}
\]
where $U$ ranges over Haar random unitaries.  Note that the outcome of the measurement consists of a partition $\lambda \vdash n$ and a unitary $U \in \C^{d \times d}$.
\end{definition}

\begin{definition}[Schur Weyl Distribution]
Given integers $n,d$ and a tuple $(\alpha_1, \dots , \alpha_d)$ with $\alpha_i \geq 0$ and $\alpha_1 + \dots + \alpha_d = 1$, the Schur-Weyl distribution $\mathrm{SW}^n(\alpha)$ is a distribution over partitions $\lambda \vdash n$ into at most $d$ parts obtained by measuring the state $\diag(\alpha_1 , \dots , \alpha_d)^{\otimes n} $ via weak Schur sampling.  When $\alpha$ is uniform, we may write $\mathrm{SW}^n_d$ instead.
\end{definition}

\section{Basic Facts}

\begin{claim}\label{claim:haar-moments}
Let $X,Y \in \C^{d \times d}$ be Hermitian matrices.  Then
\[
\E_{U}[(U^\dagger X U) \langle U^{\dagger}X U, Y \rangle  ] =  \frac{1}{d^2 - 1}\left(\norm{X}_F^2 - \frac{\tr(X)^2}{d} \right) \left( Y - \frac{\tr(Y) I}{d}\right) + \frac{\tr(X)^2 \tr(Y) I}{d^2}
\]
where the expectation is over a Haar random unitary $U$.
\end{claim}

\begin{claim}\label{claim:typical-young-tableaux1}
Let $\alpha = (\alpha_1, \dots , \alpha_d)$ be a vector of nonnegative weights summing to $1$. Then for  $\lambda \sim \mathrm{SW}^n(\alpha)$, with probability at least $1/2$,
\[
\sum_{i = 1}^d \lambda_i^2 \geq \frac{n^{1.5}}{4} \,.
\]
\end{claim}

\begin{claim}\label{claim:typical-young-tableaux2}
Let $\alpha = (\alpha_1, \dots , \alpha_d)$ be a vector of nonnegative weights summing to $1$.  Then for  $\lambda \sim \mathrm{SW}^n(\alpha)$, 
\[
\E\left[ \sum_{i = 1}^d \lambda_i^2 \right] \leq 2((\alpha_1^2 + \dots + \alpha_d^2)n^2 + n^{1.5})
\]
\end{claim}

\begin{lemma}\label{lem:sim}
    Let $0 \le \lambda \le 1$. Given $t$ copies of an unknown state $\rho$, and given a description of a density matrix $\sigma$, it is possible to simulate any measurement of $(\lambda\rho + (1 - \lambda)\sigma)^{\otimes t}$ using a measurement of $\rho^{\otimes t}$.
\end{lemma}

\section{Balanced Case}\label{sec:balanced}

We begin by presenting our algorithm for the case when $\rho$ is close to maximally mixed.  In this case, given $n$ total copies of $\rho$, our algorithm sets $t = 0.01d^2$ and measures $n/t$ copies of $\rho^{\otimes t}$, each using Keyl's POVM.  Recall that Keyl's POVM involves first obtaining a partition $\lambda$ and then obtaining a unitary $U$.  The estimator that we construct after measuring according to Keyl's POVM will be $U \diag(\lambda_1/t, \dots , \lambda_d/t)U^{\dagger}$, which we call Keyl's estimator.  We then average this estimator (with some appropriate linear rescaling) over all $n/t$ batches to construct our final estimate for $\rho$.  In Theorem~\ref{thm:shadow-balanced}, we prove that this estimator successfully solves classical shadows.

We will rely on a few of the intermediate lemmas from \cite{chen2024optimal}. We begin with a few definitions.  Note that Keyl's POVM is symmetric over the unitary in the following sense.

\begin{definition}
We say a POVM $\{ M_z \}_{z \in \calZ}$ in $\C^{d^t \times d^t}$ is copy-wise rotationally invariant if it is equivalent to 
\[
\{ U^{\otimes t} M_z (U^{\dagger})^{\otimes t} dU\}_{ z \in \calZ}
\]
where $U \in \C^{d \times d}$ is a random unitary drawn from the Haar measure. 
\end{definition}

\begin{definition}
Let $ \{ M_z \}_{z \in \calZ}$ be a POVM in $\C^{d^t \times d^t}$   that is copywise rotationally invariant.  We say a function $f:\{ M_z \}_{z \in \calZ} \rightarrow \C^{d \times d}$ is rotationally compatible with the POVM if  
\[
f(U^{\otimes t} M_z (U^{\dagger})^{\otimes t}) = U f(M_z) U^{\dagger}  
\]   
for all $z \in \calZ$ and unitary $U$.
\end{definition}

\begin{fact}
Keyl's POVM is copy-wise rotationally invariant and the estimator $(\lambda, U) \rightarrow U \diag(\lambda_1/t, \dots , \lambda_d/t)U^{\dagger}$ is rotationally compatible with Keyl's POVM.
\end{fact}
\begin{proof}
This follows immediately from Theorem~\ref{thm:schur-weyl}.
\end{proof}

When the state $\rho$ is close to maximally mixed, we can bound the mean of Keyl's estimator \\ $U \diag(\lambda_1/t, \dots , \lambda_d/t)U^{\dagger}$ as follows.

\begin{lemma}\label{lem:estimator-error}\cite{chen2024optimal}
Let $ \{ M_z \}_{z \in \calZ}$ be a POVM in in $\C^{d^t \times d^t}$   that is copywise rotationally invariant.  Let $f:\{ M_z \}_{z \in \calZ} \rightarrow \C^{d \times d}$ be a rotationally compatible estimator such that $\tr(f(M_z)) = 0$ for all $z\in\mathcal{Z}$.  Let $X = (I_d/d + E)^{\otimes t}$ and let $X' = (I_d/d)^{\otimes t} + \sum_{\mathrm{sym}} E \otimes (I_d/d)^{\otimes t-1}$.  Assume that $\norm{E}_F \leq \left(\frac{0.01}{t} \right)^4$.  Then
\[
\norm{\int_{\calZ} f(M_z) \langle X - X', M_z \rangle \,\mathrm{d}z}_{F} \leq \frac{10^5 t^2 \norm{E}_F^2}{d} \sqrt{\int_{\calZ} \frac{\norm{f(M_z)}_F^2 \tr(M_z)}{d^t} \,\mathrm{d}z} \,.
\]
\end{lemma}

\begin{corollary}\label{coro:fake-estimator-mean}\cite{chen2024optimal}
Let $\{M_{\lambda, U} \}_{\lambda, U}$ be Keyl's POVM where $\lambda$ ranges over partitions of $t$ and $U$ ranges over unitaries in $\C^{d \times d}$.  Let $X' = (I_d/d)^{\otimes t} + \sum_{\mathrm{sym}} E \otimes (I_d/d)^{\otimes t-1}$.  Then
\[
\sum_{\lambda\vdash t}\int U \diag(\lambda_1/t, \dots , \lambda_d/t)U^\dagger \cdot \langle M_{\lambda ,U}, X'  \rangle \,\mathrm{d}U = \frac{I_d}{d} + \frac{dE}{t(d^2 - 1)}  \E_{\lambda \sim \mathrm{SW}^t_d}\Bigl[\sum_{j = 1}^d \lambda_j^2 - (t^2/d) \Bigr] \,.
\]
\end{corollary}

\begin{algorithm}[h!]\caption{Shadows for balanced states}\label{alg:shadow-balanced}
    \begin{algorithmic}
        \STATE \textbf{Input:} $m$ copies of $\rho^{\otimes t}$ for some unknown quantum state $\rho \in \C^{d \times d}$
        \FOR{$j\in[m]$}
            \STATE Measure $\rho^{\otimes t}$ according to Keyl's POVM
            \STATE Let $\lambda \vdash t$ be the partition and $U$ be the unitary obtained from the measurement
            \STATE Set $D_j = U \diag(\lambda_1/t, \dots , \lambda_d/t)U^{\dagger}$
        \ENDFOR
        \STATE Compute $\theta = \E_{\lambda \sim \mathrm{SW}^t_d}[\sum_j \lambda_j^2 ] - (t^2/d)  $ \label{step:theta}
        \STATE Compute $\wh{E} =  \frac{t(d^2 - 1)}{d\theta}\left(\frac{D_1 + \dots + D_m}{m} - \frac{I_d}{d}\right)$
        \STATE \textbf{Output:} $\wh{E}$
    \end{algorithmic}
    \noindent\rule{\textwidth}{0.2pt}
\end{algorithm}


\noindent We can now prove the main theorem in the balanced case.  
\begin{theorem}\label{thm:shadow-balanced}
Let $\rho = \frac{I_d }{d} + E$ be an unknown quantum state in $\C^{d \times d}$.  Assume that $\norm{E}_F \leq \sqrt{\eps}/d^2$.  Then for any target accuracy $\eps \leq 1/d^{12}$, there is an algorithm that measures $O(1/\eps^2)$ copies of $\rho$ and  returns $\wh{E} \in \C^{d \times d}$ such that for any Hermitian matrix $O \in \C^{d \times d}$,
\[
\Pr\left[ \left\lvert\langle O,  E\rangle - \langle O, \wh{E} \rangle  \right\rvert \geq \eps \cdot \frac{\norm{O}_F}{\sqrt{d}} \right] \leq 0.1 \,.
\]
\end{theorem}
\begin{proof}
We run Algorithm~\ref{alg:shadow-balanced} with $t = 0.01d^2$ and $m = 10^6/(\eps^2 d^2)$ (note that the total number of copies used is then indeed $O(1/\eps^2)$).  The POVM in Algorithm~\ref{alg:shadow-balanced} is clearly copywise rotationally invariant and the estimator is rotationally compatible with it.  Let us use the shorthand $\{ M_z \}_{z \in \calZ}$ to denote this POVM and for $M_z$ corresponding to unitary $U$ and partition $\lambda$, we let $f(M_z) = U\diag(\lambda_1/t, \dots , \lambda_d/t)U^\dagger$.  We have
\[
\E\left[ \frac{D_1 + \dots + D_m}{m} \right] = \int_{\calZ} f(M_z) \langle M_z , (I_d/d + E)^{\otimes t} \rangle \,\mathrm{d}z
\]
where the expectation is over the randomness of the quantum measurement in Algorithm~\ref{alg:shadow-balanced}.  We can make the estimator $D_j$ have trace $0$ by simply subtracting out $I_d/d$ and adding it back at the end.  Thus, by Lemma~\ref{lem:estimator-error} and Corollary~\ref{coro:fake-estimator-mean}, recalling the definition of $\theta$ in Line~\ref{step:theta} of Algorithm~\ref{alg:shadow-balanced}, we have
\[
\begin{split}
\norm{\E\left[ \frac{D_1 + \dots + D_m}{m} \right] - \frac{I_d}{d} - \frac{d\theta E}{t(d^2 - 1)} }_{F} & \leq \frac{10^5 t^2 \norm{E}_F^2}{d} \sqrt{\int_{\calZ} \frac{\norm{f(M_z)}_F^2 \tr(M_z)}{d^t}\,\mathrm{d}z} \\ & \leq \frac{10^5t^2\norm{E}_F^2}{d} \,.
\end{split}
\]
Thus, if $\wh{E}$ is the output of Algorithm~\ref{alg:shadow-balanced}, then
\[
\norm{\E[\wh{E}] - E}_F \leq \frac{10^5t^3 \norm{E}_F^2}{\theta} \,.
\]
Next, we compute the variance of the estimator.  Note that WLOG, we can assume the observable $O$ has $\tr(O) = 0$ since our estimator $\wh{E}$ is always traceless.  We have
\[
\begin{split}
\E\left[ \langle O , \wh{E} - \E[\wh{E}] \rangle^2 \right]  & \leq \frac{d^2t^2}{m\theta^2}  \E\left[ \langle O, D_1 - \E[D_1] \rangle^2 \right]  \label{eq:boundvariance} \\
&\leq  \frac{4d^2t^2}{m\theta^2}\left( \E\left[ \left\langle O, D_1 - \frac{I_d}{d} \right \rangle^2 \right] \right) 
\\ & \leq \frac{8d^2t^2}{m\theta^2}\left( \E \left[ \E_{U}\left[ \left\langle U^\dagger O U, D_1 - \frac{I_d}{d} \right \rangle^2 \right] \right] \right) 
\\ & \leq \frac{8d^2t^2}{m\theta^2} \frac{\norm{O}_F^2 \E[\norm{D_1}_F^2]}{d^2}
\\ & = \frac{8 \norm{O}_F^2}{m\theta^2} \E_{\lambda \sim \mathrm{SW}^t(\rho)}[\sum_{j}\lambda_j^2 ]
\end{split}
\]
where in the above, we used that  $\norm{E}_F \leq 1/(10d)^8$ so 
\[
0.9(I_d/d)^{\otimes t} \preceq \rho^{\otimes t} \preceq 1.1(I_d/d)^{\otimes t} 
\]
and thus when measuring $\rho^{\otimes t}$ with Keyl's POVM (which is rotationally invariant), the rotation of the outcome is approximately uniform up to a factor of $2$.  Now by Claim~\ref{claim:typical-young-tableaux2}, we can upper bound $\E_{\lambda \sim \mathrm{SW}^t(\rho)}[\sum_{j}\lambda_j^2 ] \leq 2(\norm{\rho}_F^2t^2 + t^{1.5}) \leq 4t^{1.5}$ where recall that we set $t \le 0.01 d^2$.  Also by Claim~\ref{claim:typical-young-tableaux1}, we have $\theta \geq t^{1.5}/4$.  Thus, putting everything together, we conclude
\[
\E\left[ \langle O, \wh{E} - E \rangle^2\right] \leq 2 \cdot 10^{10} t^3 \norm{E}_F^4 \norm{O}_F^2 + \frac{(10 \norm{O}_F)^2}{m t^{1.5}} \leq \frac{\eps^2 \norm{O}_F^2}{10^2 d} \,.
\]
The desired statement then follows from Chebyshev's inequality.
\end{proof}

\section{Splitting Reduction}

As in \cite{chen2024optimal}, when $\rho$ is not balanced, we reduce to the balanced case via a splitting reduction.

\begin{definition}\label{def:split}
Let $b_1, \dots , b_d \in \mathbb{Z}_{\geq 0}$.  We define $\Split_{b_1, \dots , b_d}$ to be a linear map that sends any $M\in \C^{d\times d}$ to a square matrix with dimension $2^{b_1} + \dots + 2^{b_d}$ defined as follows.  The rows and columns of $\Split_{b_1, \dots , b_d}(M)$ are indexed by pairs $(j, s)$ where $j \in [d]$ and $s \in \{0,1 \}^{b_j}$ and these are sorted first by $j$ and then lexicographically according to $s$.  Now the entry indexed by row $(j_1, s_1)$ and column $ (j_2,s_2)$  is defined as 
\begin{itemize}
\item If $b_{j_1} \leq b_{j_2}$ then the entry is $M_{j_1j_2}/2^{b_{j_2}}$ if $s_1$ is  a prefix of $s_2$ and is  $0$ otherwise
\item If $b_{j_1} > b_{j_2}$ then the entry is $M_{j_1j_2}/2^{b_{j_1}}$ if $s_2$ is a prefix of $s_1$ and is $0$ otherwise
\end{itemize}
\end{definition}

As an example, we have the following splitting of a $2 \times 2$ matrix with $b_1 = 2, b_2 = 1$. 

\[
\Split_{2,1}\left(\begin{bmatrix} 
\textcolor{red}{a_{11}} & \textcolor{blue}{a_{12}}  \\ 
\textcolor{blue}{a_{21}} & \textcolor{green}{a_{22}} 
\end{bmatrix}\right) = \begin{bmatrix} \textcolor{red}{0.25 a_{11}} & 0 & 0 & 0 & \textcolor{blue}{0.25a_{12}} & 0 \\ 0 & \textcolor{red}{0.25 a_{11}} & 0 & 0 & \textcolor{blue}{0.25 a_{12}} & 0  \\ 0 & 0 &  \textcolor{red}{0.25a_{11}} & 0 & 0 & \textcolor{blue}{0.25a_{12}}  \\ 0 & 0 & 0 & \textcolor{red}{0.25a_{11}} & 0 & \textcolor{blue}{0.25 a_{12}} \\ \textcolor{blue}{0.25 a_{21}} & \textcolor{blue}{0.25 a_{21}} & 0 & 0 & \textcolor{green}{0.5 a_{22}} & 0 \\ 0 & 0 & \textcolor{blue}{0.25 a_{21}} & \textcolor{blue}{0.25 a_{21}} & 0 & \textcolor{green}{0.5a_{22}}  \end{bmatrix}
\]

\noindent We can also define a dual map to $\Split$.  

\begin{definition}\label{def:dsplit}
Let $b_1, \dots , b_d \in \mathbb{Z}_{\geq 0}$.  We define $\Dsplit_{b_1, \dots , b_d}$ to be a linear map that sends any $M\in \C^{d\times d}$ to a square matrix with dimension $2^{d_1} + \dots + 2^{b_d}$ defined as follows.  The rows and columns of $\Dsplit_{b_1, \dots , b_d}(M)$ are indexed by pairs $(j, s)$ where $j \in [d]$ and $s \in \{0,1 \}^{b_j}$ and these are sorted first by $j$ and then lexicographically according to $s$.  Now the entry indexed by row $(j_1, s_1)$ and column $ (j_2,s_2)$  is defined as 
\begin{itemize}
\item $M_{j_1j_2}$ if $s_1$ is  a prefix of $s_2$ and or $s_2$ is a prefix of $s_1$ and $0$ otherwise
\end{itemize}
\end{definition}    

\noindent We have the following basic properties.
\begin{claim}\label{claim:splitting-properties}
Let $b_1, \dots , b_d \in \mathbb{Z}_{\geq 0}$. We have the following statements for any $M,N \in \C^{d \times d}$:
\begin{itemize}
\item $\norm{\Split_{b_1, \dots , b_d}(M)}_F \leq \norm{M}_F $
\item $\langle \Split_{b_1, \dots , b_d}(M), \Dsplit_{b_1, \dots , b_d}(N) \rangle = \langle M,  N \rangle $
\item $\norm{\Dsplit_{b_1, \dots , b_d}(N)}_F \leq 2\sqrt{2^{b_1} + \dots + 2^{b_d}}\norm{N} $
\end{itemize}    
\end{claim}
\begin{proof}
The first two statements follow immediately from the definitions.  To verify the third, for each integer $k$, let $S_k \subseteq [d]$ denote the set of indices $j$ such that $b_j = k$.  Note that the entries of $\Dsplit_{b_1, \dots , b_d}(N)$ are  obtained by taking the entries of $N$ and duplicating each a certain number of times \--- an entry $N_{j_1j_2}$ appears $2^{\max(b_{j_1}, b_{j_2})}$ times.  Thus,
\[
\begin{split}
\norm{\Dsplit_{b_1, \dots , b_d}(N)}_F^2 & = \sum_{j_1, j_2 \in [d]} 2^{\max(b_{j_1}, b_{j_2})} N_{j_1j_2}^2 \\ &\leq   \sum_{k = 0}^{\infty} 2^k  \sum_{j_1 \in S_k \text{ or } j_2 \in S_k} N_{j_1j_2}^2  \\ &\leq \sum_{k = 0}^{\infty} 2^k (2|S_k| \norm{N}^2) \\ &= 2(2^{b_1} + \dots + 2^{b_d})\norm{N}^2
\end{split}
\]
and this gives the desired inequality.
\end{proof}

\begin{claim}[\cite{chen2024optimal}]\label{claim:simulate-measurements}
Given measurement access to $\rho^{\otimes t}$ where $\rho \in \C^{d \times d}$ is a state, $\Split_{b_1, \dots , b_d}(\rho)$ is a valid state and we can simulate measurement access to access to $\Split_{b_1, \dots , b_d}(\rho)^{\otimes t}$.    
\end{claim}

To complete the reduction, our full algorithm first obtains a rough estimate $\rho'$ of $\rho$ via tomography and then applies the splitting reduction in the eigenbasis of $\rho'$

 \begin{theorem}[\cite{guctua2020fast}]\label{thm:untentangled-tomography}
For any $\delta, \eps < 1$ and unknown state $\rho \in \C^{d \times d}$, there is an algorithm that makes unentangled measurements on $O(d^2 \log(1/\delta)/\eps^2  )$ copies of $\rho$ and with $1 - \delta$ probability outputs a state $\wh{\rho}$ such that $\norm{\rho - \wh{\rho}}_F \leq \eps$.
\end{theorem}

\begin{theorem}
Let $\rho$ be an unknown quantum state in $\C^{d \times d}$. Then for any target accuracy $\eps \leq 1/d^{12}$ and failure probability $\delta$, there is an algorithm that measures $O(\log(1/\delta)/\eps^2)$ copies of $\rho$ and stores classical information (of $O(d^2 \log(1/\delta))$ real numbers) such that given any Hermitian matrix $O \in \C^{d \times d}$, it can produce an estimate $\tau$ (from only this classical information) with
\[
\Pr\left[ \left\lvert\langle O,  \rho \rangle - \tau  \right\rvert \geq \eps \norm{O} \right] \leq \delta \,.
\]    
\end{theorem}
\begin{proof}
First, we apply Theorem~\ref{thm:untentangled-tomography} with   half of the total copies to learn a state $\rho'$ such that 
\[
\norm{\rho' - \rho}_F \leq d\eps \leq \frac{\sqrt{\eps}}{d^5} \,.
\]
Now let $U$ be the matrix that diagonalizes $\rho'$.  We will work in the $U$-basis where say  $\rho' = \diag(\lambda_1, \dots ,\lambda_d)$.  For each $j \in [d]$ let $b_j$ be the smallest nonnegative integer such that $2^{b_j} \geq d\lambda_j$.  Note that we must have $2^{b_1} + \dots + 2^{b_d} \leq 4d$.  Also, $\Split_{b_1,\dots , b_d}(\rho')$ is a diagonal matrix with all entries at most $1/d$  so $\norm{\Split_{b_1,\dots , b_d}(\rho')} \leq 1/d$.  Let $k = 2^{b_1} + \dots + 2^{b_d}$.  Now let
\[
\sigma = \frac{1}{3} \left(\frac{4I_{k}}{k} - \Split_{b_1,\dots , b_d}(\rho') \right) \,.
\]
Note that $\sigma$ is a quantum state in $\C^{k \times k}$.  Now by Claim~\ref{claim:simulate-measurements} and Lemma~\ref{lem:sim}, we simulate access to  copies of the state 
\[
\tilde{\rho} = \frac{3}{4}\sigma + \frac{1}{4} \Split_{b_1,\dots , b_d}(\rho) \,.
\]
Note that 
\[
\norm{\tilde{\rho} - \frac{I_{k}}{k}}_F = \norm{\frac{1}{4}\Split_{b_1,\dots , b_d}( \rho - \rho') }_F \leq  \frac{\sqrt{\eps}}{d^5} \,.
\]
Thus, we can apply Theorem~\ref{thm:shadow-balanced} using $O(1/\eps^2)$ copies to obtain an estimate $\wh{E}$ for $\tilde{\rho} - \frac{I_{k}}{k}$.  We get that
\[
\Pr\left[ \left\lvert \left\langle \Dsplit_{b_1, \dots , b_d}(O),  \tilde{\rho} - \frac{I_{k}}{k} \right\rangle - \left\langle \Dsplit_{b_1, \dots , b_d}(O),  \wh{E} \right\rangle  \right\rvert \geq 4\eps \norm{O}  \right] \leq 0.1 
\]
where we used Claim~\ref{claim:splitting-properties}.  When the above holds, then we set 
\[
\tau =  \langle O, \rho' \rangle + 4 \langle \Dsplit_{b_1, \dots , b_d}(O),  \wh{E} \rangle
\]
and then get
\[
\left\lvert\langle O,  \rho \rangle - \tau  \right\rvert \leq 16 \eps \norm{O} \,.
\]
To complete the proof and get $1 - \delta$ success probability, note that we can repeat the above $c = 10\log(1/\delta)$ times independently to obtain estimates $\wh{E}_1, \dots \wh{E}_c$.  Then for a given query $O$ we compute estimates $\tau_1, \dots , \tau_c$ as above and output the median of these estimates.
\end{proof}
\section{Reducing to Small $\eps$ via Random Projection}
\label{sec:reduction}

In this section, we show how to reduce to the case where $\eps \leq 1/\sqrt{d}$ by first applying a random ``projection" (we will actually use matrices with Gaussian entries so it is not technically a projection).  We show in Claim~\ref{claim:mean-calc} that we have an unbiased estimator after the projection and we bound the variance in Claim~\ref{claim:var-calc}.

\begin{definition}
For $k$ matrices $V_1 \in \C^{d \times m},  \dots , V_k  \in \C^{d \times m}$ and a matrix $M \in \C^{d \times d}$, we write $M_{V_1, \dots , V_k}$ to be the $k \times k$ matrix whose $ij$ entry is $\tr(V_i^\dagger  M V_j)$ when $i \neq j$ and is $0$ otherwise.
\end{definition} 

\begin{claim}\label{claim:mean-calc}
Let $d,m$ be integers.  Let $M,N \in \C^{d \times d}$ be traceless Hermitian matrices.  Let $V_1, \dots , V_k \in \C^{d \times m}$ be matrices with entries drawn i.i.d. from $N(0,1/d)$.  Then
\[
\E[\langle M_{V_1, \dots , V_k}, N_{V_1, \dots , V_k} \rangle ] = \frac{k(k-1)m}{d^2} \langle M, N \rangle \,.
\]
\end{claim}
\begin{proof}
Let $i,j \in [k]$ with $i \neq j$.  We have 
\[
\begin{split}
\E[ M_{V_1, \dots , V_k}[i,j] N_{V_1, \dots , V_k}[i,j] ] & = \E[ \tr(V_i^\dagger M V_j) \tr(V_i^\dagger N V_j) ] \\ &= \frac{1}{d}\E[  \langle V_i^\dagger M, V_i^\dagger N \rangle ] \\ &= \frac{m}{d^2} \langle M, N \rangle \,.
\end{split} 
\]
Now summing the above over all $i,j \in [k]$ with $i \neq j$ gives us
\[
\E[\langle M_{V_1, \dots , V_k}, N_{V_1, \dots , V_k} \rangle ]  = \frac{k(k-1)m}{d} \langle M, N \rangle \,.
\]
\end{proof}

Now we bound the variance of the above quantity.  

\begin{claim}\label{claim:var-calc}
Let $M,N \in \C^{d \times d}$ be traceless Hermitian matrices.  Let $V_1, \dots, V_k \in \C^{d \times m}$ be matrices whose entries are drawn i.i.d. from $N(0,1/d)$.    Then 
\[
\Var(\langle M_{V_1, \dots , V_k}, N_{V_1, \dots , V_k}\rangle)  \leq  6\left(\frac{m^2k^2\norm{M}_F^2 \norm{N}_F^2}{d^4} + \frac{mk^2 \langle M^2, N^2 \rangle}{d^4} + \frac{mk^3 \norm{MN}_F^2 }{d^4} \right) \,.
\]
\end{claim}
\begin{proof}
We can write
\[
\langle M_{V_1, \dots , V_k}, N_{V_1, \dots , V_k}\rangle = \sum_{i,j \in [k], i \neq j} \tr(V_i^\dagger M V_j)\tr(V_i^\dagger N V_j) \,.
\]
Define
\[
P_{ij} = \tr(V_i^\dagger M V_j)\tr(V_i^\dagger N V_j) \,,
\]
for any pairs $(i_1, j_1)$ and $(i_2, j_2)$ that are disjoint, $\Cov(P_{i_1j_1}, P_{i_2j_2}) = 0$.  Now we compute $\Cov(P_{ij_1}, P_{ij_2}) = 0$ where $i, j_1, j_2$ are all distinct.  We have
\[
\begin{split}
\E[P_{ij_1}P_{ij_2}] &=  \E_{V_i} [\E_{V_{j_1}}[\tr(V_i^\dagger M V_{j_1})\tr(V_i^\dagger N V_{j_1})]  \E_{V_{j_2}}[\tr(V_i^\dagger M V_{j_2})\tr(V_i^\dagger N V_{j_2})]] \\ & = \E_{V_i} \left[  \frac{1}{d^2}\langle V_{i}^\dagger M, V_i^\dagger N \rangle^2  \right] \\ &= \frac{1}{d^2} \left( m \E_{v}[(vMNv^\dagger)^2] + m(m-1) (\E_{v}[vMNv^\dagger])^2 \right) \\ & = \frac{m^2\tr(MN)^2 }{d^4}  + \frac{m\norm{MN}_F^2}{d^4} + \frac{m \tr((MN)^2)}{d^4}
\end{split}
\]
where the vector $v$ is a $d$-dimensional vector with entries drawn from $N(0,1/d)$. Also
\[
\begin{split}
\E[P_{ij_1}] \E[P_{ij_2}] &= \left(\E_{V_i,V_j}[\tr(V_i^\dagger M V_{j})\tr(V_i^\dagger N V_{j})]\right)^2 \\ & = \left(\E_{V_i}\left[\frac{1}{d} \langle V_i^\dagger M, V_i^\dagger N \rangle \right]\right)^2
\\ &= \frac{m^2}{d^2} (\E_{v}[vMNv^\dagger])^2
\\ & = \frac{m^2\tr(MN)^2 }{d^2}
\end{split}
\]
and thus,
\[
\Cov(P_{ij_1}, P_{ij_2}) = \frac{2m \norm{MN}_F^2}{d^4}
\]

Next, we can compute $\Var(P_{ij})$ for $i \neq j$.  We have
\[
\begin{split}
E[P_{ij}^2] &=  \E_{V_i, V_j} [\tr(V_i^\dagger M V_{j})^2\tr(V_i^\dagger N V_{j})^2] \\ &= \frac{1}{d^2}\E_{V_i}\left[  \norm{V_i^\dagger M}_F^2 \norm{V_i^\dagger N}_F^2 + 2\langle V_i^\dagger M, V_i^\dagger N \rangle^2 \right]  \\ & = \frac{m^2\norm{M}_F^2 \norm{N}_F^2 }{d^4} + \frac{2m\langle M^2, N^2 \rangle}{d^4} + \frac{2 m^2 \tr(MN)^2}{d^4} + \frac{2 m\norm{MN}_F^2}{d^4} + \frac{2 m\tr((MN)^2)}{d^4}
\end{split}
\]
and recall
\[
\E[P_{ij}]^2 = \frac{m^2\tr(MN)^2 }{d^4}
\]
so
\[
\Var(P_{ij}) = \frac{m^2 \norm{M}_F^2 \norm{N}_F^2 }{d^4} + \frac{2m\langle M^2, N^2 \rangle}{d^4} + \frac{ m^2 \tr(MN)^2}{d^4} + \frac{2 m\norm{MN}_F^2}{d^4} + + \frac{2 m\tr((MN)^2)}{d^4} \,.
\]

Thus, we can bound
\[
\Var(\langle M_{V_1, \dots , V_k}, N_{V_1, \dots , V_k}\rangle) \leq  6\left(\frac{m^2k^2\norm{M}_F^2 \norm{N}_F^2}{d^4} + \frac{mk^2 \langle M^2, N^2 \rangle}{d^4} + \frac{mk^3 \norm{MN}_F^2 }{d^4} \right)
\]
as desired.
\end{proof}

\begin{corollary}\label{coro:shadow-var-bound}
Let $M,N \in \C^{d \times d}$ be Hermitian matrices such that $\norm{M}_1 \leq 1$ and $\norm{N}_{\op} \leq 1$.   Let $k$ be a parameter and let $V_1, \dots, V_k \in \C^{d \times 2d/k}$ be matrices whose entries are drawn i.i.d. from $N(0,1/d)$.    Then for any parameter $\gamma > 0$,
\[
\Pr\left[ \left\lvert \langle M, N \rangle - \frac{d}{2(k-1)} \langle M_{V_1, \dots , V_k}, N_{V_1, \dots , V_k}\rangle \right\rvert \geq \gamma \right] \leq \frac{20 d}{k^2 \gamma^2} \,.
\]     
\end{corollary}
\begin{proof}
This follows from combining Claim~\ref{claim:mean-calc} and Claim~\ref{claim:var-calc} and applying Chebyshev's inequality.    
\end{proof}

Before we can complete the reduction, we need a few basic matrix concentration bounds.  The following statements follow from standard matrix concentration inequalities.

\begin{claim}\label{claim:matrix-chernoff}
Let $V_1, \dots, V_k \in \C^{d \times 2d/k}$ be matrices whose entries are drawn i.i.d. from $N(0,1/d)$.  Then with probability $0.99$,
\[
0.1 I_d \preceq \sum_{i = 1}^k V_i V_i^\dagger \preceq 10 I_d \,.
\]
Also let $N \in \C^{d \times d}$ be a fixed matrix with $\norm{N}_{\op} \leq 1$.  Then with probability $0.99$
\[
\norm{N_{V_1, \dots , V_k}}_{\op} \leq 10 \,.
\]
\end{claim}

\noindent
Before we formalize the reduction, we first need the following definition:

\begin{definition}
We say that an algorithm solves classical shadows with parameters $d, \eps$ with probability $p$ using $f(d,\eps , p)$ copies if for an arbitrary unknown state $\rho$, the algorithm makes measurements on $f(d,\eps , p)$ copies of $\rho$ and stores classical information such that for any observable $O \in \C^{d \times d}$, the algorithm can access only the classical information and produce an estimate $\tau$ such that with probability $p$, 
\[
|\langle O, \rho \rangle - \tau | \leq \eps \norm{O} \,.
\]
\end{definition}

\noindent With this, we can now show:

\begin{theorem}\label{lem:reduce-dimension}
Assume we are given parameters $d,\eps, k$ such that  $d \geq k \geq 100\sqrt{d}/\eps$. Then if there is an algorithm for solving classical shadows with probability $0.9$ for parameters $k, 0.1 \eps k / d$ using $f(k, 0.01 \eps k / d, 0.9) $ samples, then for any $\delta > 0$, there is also an algorithm for solving classical shadows with parameters $d,\eps$ that succeeds with probability $1 - \delta$ and uses  $f(d,\eps , 1- \delta) \leq O(f(k, 0.01 \eps k / d, 0.9) + 1/\eps^2 ) \cdot \log(1/\delta) $ samples.
\end{theorem}

\begin{proof}
It suffices to prove the statement with $\delta = 1/3$ and then we can simply take $\log(1/\delta)$ independent runs of the algorithm and output the median of the estimates.

Now draw $V_1, \dots , V_k \in \C^{d \times 2d/k}$ with i.i.d. entries drawn from $N(0,1/d)$.  Recall by Claim~\ref{claim:matrix-chernoff}, with probability $0.99$, 
\begin{equation}\label{eq:condition-number}
0.1I_d \preceq \sum_{i = 1}^k V_i V_i^\dagger \preceq 10 I_d  \,.
\end{equation}
Assuming the above holds, we define
\[
Y = I_d - 0.1 \sum_{i = 1}^k V_i V_i^\dagger
\]
and construct the following quantum channel.  We map a matrix $M \in \C^{d \times d}$ to a block diagonal matrix with a $d \times d$ block consisting of $Y^{1/2} M Y^{1/2}$ and a $k \times k$ block with entries given by $0.1 \tr(V_i^\dagger M V_j) $ for all $i,j \in [k]$ (this is similar to the matrix $M_{V_1, \dots , V_k}$ except the diagonal is included as well).  It is clear that this map is completely positive and trace preserving.   We run all copies of $\rho$ through this channel and then measure them according to the POVM $(\Pi , I_{d+k} - \Pi)$ where $\Pi$ is the projector onto the $k \times k$ block.  We keep all of the copies for which the measurement outcome is $\Pi$ and discard the rest.  Let
\[
\alpha = \tr(V_1^\dagger \rho V_1) + \dots + \tr(V_k^\dagger \rho V_k) \,.
\]
Note that by \eqref{eq:condition-number}, $\alpha \geq 0.1$.  Let $\beta$ be the fraction of samples that we actually keep.  Since we have more than $O(1/\eps^2)$ samples, with  $0.99$ probability, $0.1 \alpha - 0.001\eps \leq \beta \leq 0.1 \alpha + 0.001\eps$.  Note that all of these copies are now in the state
\[
\rho' = \frac{\rho_{V_1, \dots , V_k} + \diag(\tr(V_1^\dagger \rho V_1), \dots , \tr(V_k^\dagger \rho V_k))}{\tr(V_1^\dagger \rho V_1) + \dots + \tr(V_k^\dagger \rho V_k)} \,.
\]

We now run the classical shadows algorithm with parameters $k, 0.01 \eps k / d$ on this $k \times k$ matrix.  As long as enough samples are kept in the previous step i.e. $\beta \geq 0.1\alpha - 0.001\eps$, we have enough copies remaining to run this algorithm.  Given a query, $O$, we construct the matrix $O_{V_1, \dots , V_k}$ and then compute an estimate $\tau'$ for $\langle \rho' , O_{V_1, \dots , V_k} \rangle$.

WLOG assume $\norm{O}_{\op} \leq 1$.  Then by the assumption about the classical shadows algorithm, and as long as $\norm{O_{V_1, \dots , V_k}}_{\op} \leq 10$ (which happens with $0.99$ probability by Claim~\ref{claim:matrix-chernoff}), our estimate $\tau'$ satisfies
\[
|\tau' - \langle \rho' , O_{V_1, \dots , V_k} \rangle| \leq \frac{0.1\eps k}{d} 
\]
with probability at least $0.9$.  Now we output the estimate  
\[
\tau = \frac{5d\beta }{(k-1)} \tau' \,.
\]
Assuming that the previous inequality holds,
\[
\begin{split}
\left\lvert \tau - \frac{10\beta}{\alpha} \langle \rho , O \rangle \right\rvert  \leq \left\lvert \frac{10\beta}{\alpha}  \langle  \rho , O \rangle - \frac{5d\beta}{(k-1)}  \langle \rho' , O_{V_1, \dots , V_k} \rangle \right\rvert + 0.6\eps \\ = \frac{10\beta}{\alpha} \left\lvert  \langle \rho , O \rangle - \frac{d}{2(k-1)}  \langle \rho_{V_1, \dots , V_k} , O_{V_1, \dots , V_k} \rangle \right\rvert + 0.6\eps
\end{split}
\]
where the last step uses that $O_{V_1, \dots , V_k}$ is $0$ on the diagonal by definition.  Now by Corollary~\ref{coro:shadow-var-bound}, with probability $0.9$, the above quantity is at most $0.8\eps$.  Also assuming $0.1 \alpha - 0.001\eps \leq \beta \leq 0.1 \alpha + 0.001\eps$, we have
\[
\left\lvert \frac{10\beta}{\alpha} - 1 \right\rvert \leq 0.1\eps \,.
\]
This immediately implies that our estimate $\tau$ has $| \tau - \langle \rho, O \rangle| \leq \eps$.  Overall, combining the failure probabilities over all of the steps, the total failure probability is less than $1/3$, so with $2/3$ probability, the estimate $\tau$ is $\eps$-accurate and we are done.
\end{proof}

\end{document}